%% file: computingeigs-arxiv2.tex
\documentclass{article}[12pt]
\input{macros}

\usepackage{fullpage}
\usepackage[mathscr]{euscript}

\renewcommand{\defeq}{:=}

\title{Accelerated Newton Iteration:
 Roots of Black Box Polynomials\\ and Matrix Eigenvalues}

\author{Anand Louis\thanks{Supported by the Simons Collaboration on Algorithms and Geometry.}
\\Princeton University \\alouis@princeton.edu \and 
	Santosh S. Vempala\thanks{Supported in part by NSF awards CCF-1217793 and EAGER-1555447.}
 \\ Georgia Tech \\ vempala@gatech.edu}

\newcommand{\charp}{f} 
\newcommand{\mom}{g}  

\newcommand*\diff{\mathop{}\!\mathrm{d}}
\newcommand{\dx}{{\sf dx}}
\newcommand{\tg}{\tilde{g}}
\newcommand{\ut}{u^t}
\newcommand{\tut}{\tilde{u}^t}
\newcommand{\ralsymb}{\mathscr{R}}
\newcommand{\eprime}{\e'}
\def\eps{\epsilon}

\newcommand{\ftime}[1]{\cT_{\charp}\paren{#1}}
\newcommand{\bitcomp}[1]{\cB\paren{#1}}

\date{}

\begin{document}

\begin{titlepage}

\maketitle

\begin{abstract}
We study the problem of computing the largest root of a real rooted polynomial $p(x)$ to within error $\e$ given only 
black box access to it, i.e., for any $x \in \R$, the algorithm can query an oracle for the
value of $p(x)$, but the algorithm is not allowed access to the coefficients of $p(x)$. 
A folklore result for this problem is that the largest root of a polynomial can be
computed in $\bigo{n \log (1/\e)}$ polynomial queries 
using the Newton iteration.
We give a simple algorithm that queries the oracle at only $\bigo{{\log n} \log(1/\e)}$ points,
where $n$ is the degree of the polynomial.
Our algorithm is based on a novel approach for accelerating the Newton method by using higher derivatives.

As a special case, we consider the problem of computing the top eigenvalue of a symmetric matrix  in $\Q^{n \times n}$ to within error $\e$ in time polynomial in the input description, i.e., the number of bits to describe the matrix and $\log(1/\e)$. Well-known methods such as the power iteration and Lanczos iteration incur running time  
polynomial in $1/\e$, while Gaussian elimination takes $\Omega(n^4)$ bit operations.
As a corollary of our main result, 
we obtain a $\tbigo{n^{\omega}  \log^2 ( \fnorm{A}/\e)}$  bit complexity algorithm to compute
the top eigenvalue of the matrix $A$ or to check if it is approximately PSD ($A \succeq -\eps I$).
\end{abstract}

\end{titlepage}

\section{Introduction}

Computing the roots of a polynomial is a fundamental algorithmic 
problem. 
According to the folklore Abel-Ruffini theorem,
polynomials of degree five or higher do not have any algebraic solution in general,
and the roots of polynomials can be irrational. Therefore, the roots of a polynomial
can only be computed only to some desired precision. 
The classical Newton's method (also known as the Newton-Raphson method) is an 
iterative method to compute the roots of a real rooted polynomial.
Starting with an initial upper bound $x^0 \in \R$ on the largest root of a polynomial
$\charp(x)$ of degree $n$, the Newton's method recursively computes better estimates
to the largest root as follows
\[ x^{t+1} \defeq x^t - \frac{ \charp(x^t)  }{\charp'(x^t)}  \mper \] 
A folklore result is that after $\bigo{n \log (x^0/\e)}$ iterations, 
$x^t$ will be $\e$-close to the the largest root of the polynomial.

We study the problem of computing the largest root of a real rooted polynomial $p(x)$ given only 
blackbox access to it, i.e., for any $x \in \R$, the algorithm can query an oracle for the
value of $p(x)$, but the algorithm is not allowed access to the coefficients of $p(x)$. 
This model is useful when the polynomial is represented implicitly, and each evaluation
of the polynomial is computationally expensive. An important example is the characteristic polynomial,
say $\charp(x)$, of a matrix, say $A$; each evaluation of $\charp(x)$ amounts to computing
the determinant of the matrix $\paren{A - xI}$. More generally,   
equations involving determinants of polynomial matrices  fall into this category.
A slightly modified Newton's method can be used to compute the largest root of a polynomial 
using $\bigo{n \log (x^0/\e)}$ black box queries; we review this in \prettyref{sec:Newton}.

\paragraph{Computational Model.}
The two most common ways of measuring the time complexity of an algorithm are its 
{\em arithmetic} complexity, and its {\em boolean} or {\em bit} complexity.
Arithmetic complexity counts the number of basic arithmetic operations (i.e. addition, subtraction,
multiplication, division) required to execute an algorithm, whereas boolean/bit complexity
counts the number of bit operations required to execute the algorithm.
For most algorithms for combinatorial optimization problems, these two notions of time complexity are 
roughly the same with arithmetic operations being done on $\bigo{\log n}$-bit numbers.
However, for many numerical algorithms, they can differ vastly.
For example, 
Gaussian elimination is usually said to take $\bigo{n^3}$ time, but this usually refers to the number of 
arithmatic operations, and if done naively, the intermediate bit complexity can be exponential in $n$
\cite{b66,f77}.
However, using a more careful variant of 
Gaussian elimination due to Edmonds \cite{e67}, the bit complexity is known to be $\tbigo{n^4}$ 
(see also \cite{bar68,dix82,s98}). 
In this paper, we will be bounding the bit complexity of our algorithms.

\paragraph{Matrix Eigenvalues.}
The eigenvalues of a matrix are also the roots of its characteristic polynomial.
For a matrix $A \in \Z^{n \times n}$, with eigenvalues $\lambda_1 \geq \lambda_2 \geq  \ldots \geq \lambda_n$ let 
$\charp$ denote its characteristic polynomial, i.e.,
\[ \charp(x) \defeq \det \paren{x I - A } = \Pi_{i \in [n]} (x - \lambda_i) \mper \] 
We note that the algorithms for computing the roots of a polynomial are not directly useful 
for computing the eigenvalues of a matrix, as
computing the characteristic polynomial of a matrix is a computationally non-trivial task; 
the current best algorithm to compute the characteristic polynomial of a matrix 
is due to Kaltofen and Villard \cite{kv05b} achieving bit complexity $\tbigo{n^{2.697}}$.

It is well-known that the top eigenvalue can be approximated to any desired accuracy 
$\e$ in polynomial time; indeed all $n$ eigenvalues of an $n \times n$ matrix or all singular values 
of an $m \times n$ matrix can be 
computed; for the latter, it is common to state that the asymptotic 
time complexity is $ \bigo{ \min \set{ mn^2, nm^2} }$. However, this 
bound does not reflect the dependence on $\e$. Standard iterative methods for SVD 
can take time that grows polynomially with $1/\e$, which is undesirable.
The most popular algorithm for computing the top eigenvalue of a matrix is the 
{\em Power Iteration} algorithm \cite{mp29} having a running time
bound of $\bigo{n^2 (\log n)/\e}$, and this bound is tight when the matrix has its top few 
eigenvalues close to each other.

We seek algorithms for computing the top eigenvalue of a matrix whose running time is polynomial  
in the input description, i.e., the number of bits to describe the matrix $A$ and 
the parameter $\e$, the latter being $\log (1/\e)$ bits.

\subsection{Our results}
In this paper
 we study an alternative approach, inspired by the classical Newton iteration for finding roots of polynomials. 
Applying the Newton iteration (see \prettyref{sec:Newton}), we see that the iterates converge 
within $\tbigo{n}$ iterations.  
Can we do better than this? Our main idea is to accelerate the Newton iteration using higher derivatives 
of the polynomial. The standard generalization 
to Householder methods via higher derivatives \cite{h70,or70} does not give any significant benefit. In \prettyref{sec:algo}, 
we give an new iteration based on higher derivatives that converges faster, 
yielding our following main result.
The complete algorithm is described in \prettyref{fig:accnewton}.
\begin{theorem}
\label{thm:poly-blackbox}
Given black-box access to a monic real rooted polynomial $\charp$ of degree $n$, an upper bound $\gamma$ 
on the absolute value of its roots,
and an error parameter $\e \in (0,1/2]$, there exists a deterministic algorithm that queries $\charp$
at $\bigo{\log n \log (\gamma/\e) }$ locations, 
each having precision $\bigo{ \log n \log (n\gamma/\e)}$ bits, 
and outputs an $x \in \Q$ satisfying
$\lambda_1 \leq x \leq \lambda_1 + \e$,
where $\lambda_1$ is the largest root of $\charp$. 
\end{theorem}

Computing the determinant of an integer matrix has asymptotic bit complexity 
$\bigo{n^{\omega} \log^2 n \log (\fnorm{A})}$ for any integer matrix $A$ \cite{s05}. 
Using this determinant algorithm as a black box,
we get the following result for computing the eigenvalues of matrices.
\begin{theorem}
\label{thm:main}
Given a symmetric matrix $A \in \Q^{n \times n}$, and a
parameter $\e \in (0,1/2]$, there exists a Las Vegas algorithm having bit complexity 
$\tbigo{n^{\omega} \log^2 \paren{\fnorm{A}/\e}}$ 
that outputs an $x \in \Q$ satisfying
$\lambda_1 \leq x \le \lambda_1 + \e$,
where $\lambda_1$ is the largest eigenvalue of $A$. 
\end{theorem}
A closely related problem is that of determining whether a given matrix is PSD. This 
problem arises naturally in the context of solving SDPs. \prettyref{thm:main} yields
an algorithm to check if a matrix is PSD.
\begin{corollary}[Corollary to \prettyref{thm:main}]
Given a symmetric matrix $A \in \Q^{n \times n}$, and a
parameter $\e \in (0,1/2]$, there exists a Las Vegas algorithm having bit complexity 
$\tbigo{n^{\omega} \log^2 \paren{\fnorm{A}/\e}}$ 
to check if $A \succeq -\e I$.
\end{corollary}

\subsection{Related work}
A folklore result about the Newton's iteration is that it has quadratic local convergence, i.e.,
 if the initial estimate $x^0$ is ``close'' to a root $a$ of the function $f$, 
then $(x^{t+1} - a)$ is roughly $\bigo{\paren{x^t - a}^2}$.  
Kou \etal \cite{klw06} gave a modification to the Newton's method that has local cubic convergence.
Gerlach \cite{g94} (see also \cite{fp96,kkz97,kg00}) gave a way to modify the function $f$ 
to obtain a function $F_m$ (where $m \in \Z^+$ is a parameter) such that 
the Newton's Method applied to $F_m$ will yield local convergence of order $m$
(the complexity of the computation of $F_m$ increases with $m$).
Ezquerro and Hern\'{a}ndez \cite{eh99}, and 
Guti{\'e}rrez and Hern\'{a}ndez \cite{gh01} gave an acceleration of the Newton's method 
based on the convexity of the function.
Many other modifications of the Newton's method have been explored in the literature,
for e.g. see \cite{ow08,lr08}, etc. None of these improve the asymptotic worst-case complexity of root-finding.

\paragraph{Explicit polynomials.} A related problem is to compute the roots of an explicit polynomial of degree $n$, say $p(x)$, to within error $\e$.
Pan \cite{p96} gave an algorithm to compute all the roots of an explicit polynomial
using $\tbigo{n}$ arithmetic operations;
the bit complexity of this algorithm is bounded by $\tbigo{n^3}$.  
We refer the reader to a survey by Pan \cite{p97} for a comprehensive discussion on 
algorithmic results for this problem.
We note that this model is different from the blackbox model that we study;
in the blackbox model of a polynomial $p$, the algorithm can query an oracle for the
value of $p(x)$ for any $x \in \R$, but the algorithm is not allowed access to the coefficients of $p(x)$. 

\paragraph{Other Iterations.}
The most popular algorithm for computing the top eigenvalue of a matrix is the 
{\em Power Iteration} algorithm \cite{mp29}, where for a symmetric matrix $A \in \Q^{n \times n}$,
we start with a random vector $X^0 \in \Q^n$ and recursively define $X^t$ as
\[ X^{t+1} \defeq \frac{ A X^t }{ \norm{A X^t}  } \qquad \forall t \in \Z \mper  \] 
It is easy to show that $X_t$ converges to the eigenvector corresponding the $\lambda_1$, the largest 
eigenvalue of $A$, and after $t = (\log n)/\e$ iterations, the Raleigh 
quotient\footnote{The Rayleigh quotient of a vector $X$ w.r.t. a matrix $A$ is defined as 
$ \ralsymb_A(X) \defeq (X^T A X) / (X^T X) \mper $ } 
of $X_t$  is an $\e$-approximation to $\lambda_1$. Therefore, this gives a running time
bound of $\bigo{n^2 (\log n)/\e}$, and this bound is tight when the matrix has its top few 
eigenvalues close to each other. 
Other methods such as the the Jacobi Method \cite{r71}, the Householder method \cite{h70}, etc. 
have worst case running time $\bigo{n^3}$, whereas methods such as 
{\em Lanczos} algorithm \cite{l50}, the {\em Arnoldi iteration} \cite{a51}, etc.
have a polynomial dependance on $1/\e$ in the running time.
We refer the reader to \cite{ptvf92} for a comprehensive discussion.

\paragraph{Matrix Eigenvalues.}
An algorithm for computing the largest eigenvalue of a matrix 
can be obtained by checking {\em PSDness} of a sequence of matrices, namely, a binary search 
for $x$ s.t. $xI-A$ is PSD. Checking whether a matrix is PSD can be done using Gaussian elimination 
in $\tbigo{n^4}$ bit operations \cite{e67}.

Algorithms due to \cite{pc99} (see also \cite{nh13}) compute all the eigenvalues of a matrix
in $\tbigo{n^3}$ arithmetic operations. \cite{ddh07} gave an algorithm to compute the eigenvalues
in $\tbigo{n^{\omega}}$ arithmetic operations.
Independently and concurrently, Ben-Or and Eldar \cite{be15} gave an algorithm
having boolean complexity $\tbigo{n^{\omega + \nu}}$ for any $\nu >0$,  
to compute all the eigenvalues of a matrix.
Faster methods are known for special matrices 
such as diagonally dominant matrices (of which Laplacians are an important special case), but their dependence 
on $1/\e$ is again polynomial \cite{v13}.

\subsection{Preliminaries}
\label{sec:simpleass}
\begin{assumption}
\label{ass:root1}
Given a real rooted polynomial $f$ of degree $n$ and an upper bound $a$ on the absolute value of its roots, 
the roots of the polynomial $f(4ax)/(4a)^n$ lie in $[-1/4,1/4]$
and the roots of the polynomial $f(4ax -1/4 )/(4a)^n$ lie in $[0,1/2]$. 
Therefore, we can assume without loss of generality that the given polynomial has all its roots in the 
range $[0,1/2]$.
Similarly, for a symmetric matrix $A$, $0 \preceq I/4 + A/(4 \fnorm{A}) \preceq I/2$. 
Note that in both these cases, we will need to scale the error parameter $\e$ accordingly;
since our algorithms will only have a logarithmic dependance on $1/\e$, this scaling
will not be a problem.
\end{assumption}

\paragraph{Notation.}
We will use $\charp_A(x)$ to denote the characteristic polynomial of a matrix $A$;
we will drop the subscript $A$ whenever the matrix is clear from the context.
For an $x \in \R$, we use $\bitcomp{x}$ to denote the bit complexity of $x$, i.e., the number of 
bits need to represent $x$. 
For a function $g$, we use $\tbigo{g}$ to denote $\bigo{g \log^c g}$ for absolute constants $c$.
For a function $g$, we use $g^{(k)}(x)$ to denote its $k^{th}$ derivative w.r.t. $x$. 

\section{The basic Newton iteration}
\label{sec:Newton}
For finding the root of a polynomial function $\charp(\cdot) : \R \to \R$, the basic Newton iteration is the following:
initialize $x^0 = 1$, 
and then
\[ x^{t+1} \defeq x^t - \frac{\charp(x^t)}{\charp'(x^t)} \mper \]
If $x^0 \geq \lambda_1$, then this iteration maintains $x^t \geq \lambda_1\ \forall t$ and 
reduces $x^t-\lambda_1$ by a factor of at least $\paren{1-\frac{1}{n}}$ from the following observation.
\begin{proposition}
For any $t$, the Newton iterate $x^t$ satisfies $x^t  \geq \lambda_1$ and 
\[ x^t -\lambda_1 \geq \frac{\charp(x^t)}{\charp'(x^t)} \geq \frac{x^t - \lambda_1}{n} \]
\end{proposition}
\begin{proof}
Since $\charp(x) = \Pi_{i \in [n]}(x - \lambda_i)$ we have
\[ \frac{\charp(x)}{\charp'(x)}  = \frac{1}{ \sum_{i \in [n]}  \frac{1}{x - \lambda_i}} \qquad \textrm{and } \qquad
x - \lambda_1 \geq \frac{1}{ \sum_{i \in [n]}  \frac{1}{x - \lambda_i}} \geq \frac{x^t - \lambda_1}{n} \mper   \]
\end{proof}
Along with the next elementary lemma, we get a bound of $\bigo{n \log \paren{1/\e}}$ on the number of iterations
needed for $x^t$ to be $\e$ close to $\lambda_1$. 
\begin{lemma}
\label{lem:noofits1}
Let $x^0, x^1, \ldots$ be iterates satisfying $x^0 \geq \lambda_1$ and  
\[  x^{t + 1} \leq  x^t -  \frac{x^t - \lambda_1 }{q(n) }. \]
Then for all $t\geq q(n) \ln (1/\e)$, we have 
$  0 \leq x^t - \lambda_1 \leq \e \mper   $
\end{lemma}

\begin{proof}
Suppose the condition is satisfied. Then, 
\[   \frac{ x^{t+1} - \lambda_1  }{ x^t - \lambda_1 } \leq 1 - \frac{1}{ q(n)} \mper  \]
Therefore,
\[  \paren{x^t - \lambda_1} \leq   \paren{1 - \frac{1}{ q(n)}}^t \paren{x^0 - \lambda_1} 
\leq \paren{1 - \frac{1}{q(n)}}^t  \mper \]
Hence, for all $t \geq q(n) \log \paren{1/\e}$, we have  
$ 0 \leq  x^t - \lambda_1 \leq \e \mper  $
\end{proof}

This leaves the task of computing $\charp'(x)$. We can simply use the 
approximation $(\charp(x+\delta)-\charp(x))/\delta$ for a suitably small
$\delta$. Thus the modified iteration which only needs evaluation of $\charp$ (i.e., determinant computations
when $\charp$ is the characteristic polynomial of a matrix), 
is the following: initialize $x^0 = 1$, and then 
\[ x^{t+1} \defeq x^t - \frac{\delta}{2}\frac{\charp(x^t)}{\charp(x^t+\delta)-\charp(x^t)}  \]
with $\delta = \e^2$.
 
When $\charp(\cdot)$ is the characteristic polynomial of a matrix $A$, 
evaluation of $\charp(x)$ reduces to computing $\det(A - xI)$ which can be done 
using \prettyref{thm:fast-determinant}. This gives an 
overall bit complexity of $\tbigo{n^{\omega+1}}$ for computing the top eigenvalue.

\section{Accelerating the Newton iteration}
\label{sec:algo}

To see the main idea, consider the following family of functions.
For any $k \in \Z$, define 
\[ \mom_k(x) \defeq \sum_{i \in [n]}  \frac{1}{ \paren{x - \lambda_i}^k  } \mper \]
We define the $k$'th order iteration to be 
\begin{mybox}
\begin{equation}
\label{eq:kth-order}
x^{t+1} \defeq x^t - \frac{1}{n^{1/k}}\frac{\mom_{k-1}(x^t)}{\mom_k(x^t)}
\end{equation}
\end{mybox}

Note that $\mom_1(x) = \charp'(x)/\charp(x)$ and for $k=1$ we get the Newton iteration, as $\mom_0(x) = n$.
Viewing the $\mom_k(x)$ as the $k$'th moment of the vector 
$\paren{\frac{1}{x-\lambda_1}, \frac{1}{x-\lambda_2}, \ldots, \frac{1}{x-\lambda_n}}$, we can use the following basic norm inequality.

\begin{lemma}
\label{lem:knorm}
For any vector $X \in \R^n$,
\[  \frac{1}{n^{1/k}} \norm{X}_{k-1}^{k-1} \normi{X}   \leq  \norm{X}_k^k \leq \norm{X}_{k-1}^{k-1} \normi{X}   \mper  \]
\end{lemma}
\begin{proof}
Using Holder's Inequality, we get 
\begin{equation}
\label{eq:helper1}
\norm{X}_{k-1}^{k-1}  = \sum_{i} X(i)^{k-1} \leq \paren{\sum_i X(i)^k}^{\frac{k-1}{k}} n^{  \frac{1}{k}} 
	= \norm{X}_k^{k-1} n^{\frac{1}{k} }
\mper 
\end{equation}
By the monotonicity of norms, we have
$\normi{X} \leq \norm{X}_k$. 
Therefore,
\[ \frac{1}{n^{1/k}} \norm{X}_{k-1}^{k-1} \normi{X}   \leq  \norm{X}_k^k \mper    \]
Next,
\[ \norm{X}_k^k = \sum_i X(i)^k \leq \normi{X} \sum_{i} X(i)^{k-1} =   \norm{X}_{k-1}^{k-1} \normi{X}   \mper \]
\end{proof}

The lemma implies that the distance to $\lambda_1$ shrinks by a factor of $(1-1/n^{1/k})$ in 
each iteration, thereby needing only $\tbigo{n^{1/k}}$ iterations in total. 

This brings us to question of how to implement the iteration, i.e., how to 
compute $\mom_k(x)$? We first note that these can be rewritten in terms of higher
 derivatives of the polynomial. Let $\mom_k^{(i)}(x)$ be the $i$'th derivative of $\mom_k(x)$.
\begin{lemma}
\label{lem:momprime}
For any $k \in \Z$,
\[ \mom_k'(x) = - k\, \mom_{k+1}(x)  \mper  \]
\[ \mom_k^{(i)}(x) = (-1)^i \mom_{k+i}(x) \prod_{j=0}^{i-1} (k+j) \mper \] 
\end{lemma}

\begin{proof}
\[ \mom_k'(x) = \frac{\diff}{\dx} \paren{  \sum_{i=1}^n \frac{1}{ (x-\lambda_i)^k }} = 
\sum_{i=1}^n   -k \cdot \frac{1}{ (x-\lambda_i)^{k+1} }  = -k\, \mom_{k+1} (x) \mper  \]
The second part is similar.
\end{proof}

Therefore the iteration \prettyref{eq:kth-order} is simply a ratio of higher derivatives of the polynomial. 
In the complete algorithm below (\prettyref{fig:accnewton}), which only needs evaluations of $\charp(\cdot)$, 
we approximate $\mom_l(x)$ using finite differences. The folklore finite difference method says that for any
function $f: \R \to \R$, its $k^{th}$ derivative can be estimated using  
\[ \frac{1}{\delta^{k}} \paren{ \sum_{i = 0}^{k} (-1)^i \binom{k}{i}  f \paren{x + (k-i)\delta }}  \]
for small enough $\delta$.
We prove this rigorously in our setting in \prettyref{lem:approx-der}.

\begin{figure}[ht]
\begin{tabularx}{\columnwidth}{|X|}
\hline
\begin{algorithm}[Higher-order Newton Iteration]
\label{alg:faster-newton}~

{\bf Input:}  
A real rooted monic polynomial $\charp$ of degree $n$ such that all its roots lie in $[0,1/2]$ (\prettyref{ass:root1}),
error parameter $\e \in (0,1/2]$, iteration depth $k$. 

{\bf Output:} A real number $\lambda$ satisfying $ 0 \leq \lambda - \lambda_1 \leq \e$, 
	where $\lambda_1$ is the largest root of $\charp$.

\begin{enumerate}
\item 
\label{step:deltainit}
Initialize $x^0 = 1$, 
\[ \eprime \defeq \frac{\e}{8 n^{1/k}}, \qquad \qquad \delta \defeq \frac{\eprime}{16 (2e)^k k}, \qquad \qquad
	\delta' \defeq \delta^{k+1},  \qquad \textrm{and} \qquad 
	\alpha \defeq \frac{\delta' \eprime^2}{2n^2} \mper \]

\item Repeat for $t = 1$ to $ \ceil{ 16 n^{1/k} \log(1/\e) }$ iterations:
	\begin{enumerate}
		\item Compute $\tg_k(x^t)$ as follows.
		\[ \tg_1(x) \defeq  \frac{1}{\charp(x)} \paren{\frac{\charp(x+\alpha) - \charp(x) }{\alpha}} 
		 \]
		and
		\begin{equation}
		\label{eq:tgdef} 
		\tg_{k+1}(x) \defeq \frac{(-1)^k}{k!} \paren{ \frac{1}{\delta^{k}} 
			\paren{ \sum_{i = 0}^{k} (-1)^i \binom{k}{i}  \tg_1 \paren{x + (k-i)\delta }} } \mper 
		\end{equation}
		\item 
		\label{step:faster-newton-update}
		Compute the update
		\[ \ut \defeq \frac{1}{4n^{1/k}} \frac{ \tg_{k-1}(x^t) }{ \tg_k(x^t) } \mper  \]
				
		\item  \label{step:stop}
		If  $\ut \leq \eprime$, then {\sf Stop} and output $x^t$.
		\item \label{step:update} 
		If $\ut > \eprime$, then round down $\ut$ to an accuracy of $\eprime/n$ to get $\tut$ and
		set $ x^{t+1} \defeq x^t -  \tut$.
	\end{enumerate}
\item Output $x^{t}$. \label{step:endofalg}
\end{enumerate}

\end{algorithm}
\\
\hline 
\end{tabularx}
\caption{The Accelerated Newton Algorithm}
\label{fig:accnewton}
\end{figure}

\paragraph{Discussion.}
While it is desirable for $x^t$ to be very close to $\lambda_1$, 
for $\tg_k(x^t)$ to be a good approximation of $\mom_k(x^t)$, we need $\alpha$ and $\delta$ to be
sufficiently smaller than $x^t - \lambda_1$. Equivalently, we need a way to detect when  
$x^t$ gets ``very close'' to $\lambda_1$; \prettyref{step:stop} does this for us 
(\prettyref{lem:noofits2}). 
We also want to keep the bit complexity of $x^t$ bounded; 
\prettyref{step:update} ensures this by retaining only a small number of the most significant bits
of $\ut$.

The analysis of the algorithm can be summarised as follows. 
\begin{theorem}
\label{thm:analysis}
Given a monic real rooted polynomial $\charp : \R \to \R$ of degree $n$, having all its roots in
$[0,1/2]$, \prettyref{alg:faster-newton}
outputs a $\lambda$ satisfying
\[ 0 \leq \lambda - \lambda_1 \leq \e  \]
while evaluating $\charp$ at $\bigo{k n^{1/k} \log (1/\e)}$ locations on $\R$.
Moreover, given access to a blackbox subroutine to evaluate $\charp(x)$ which runs in 
time \footnote{We assume that $\ftime{cn} = \bigo{\ftime{n}}$ for absolute constants $c$, and that
 $\ftime{n_1} \leq {\ftime{n_2}}$ if $n_1 \leq n_2$.}
$\ftime{\bitcomp{x}}$ ,
\prettyref{alg:faster-newton} has overall time complexity
$\tbigo{k n^{1/k} \log (1/\e) \ftime{k^2 + k\log (n/\e) }}$.
\end{theorem}

\subsection{Analysis}

We start with a simple fact about the derivaties of polynomials.
\begin{fact}
\label{fact:charpderivative}
For a degree $n$ polynomial $\charp(x) = \Pi_{i \in [n]}(x - \lambda_i)$, and for $k \in \N$, 
$k \leq n$, we have
\[ \charp^{(k)}(x) = k! \charp(x) \paren{ \sum_{ \substack{ S \subset [n] \\ \Abs{S} = k  } } 
		\Pi_{i \in S} \frac{1}{x - \lambda_i} } \mper  \]
\end{fact}

\begin{proof}
We prove this by induction on $k$. 
For $k = 1$, this is true.  
We assume that this statement holds for $k = l$ ($l<n$), and show that it holds for $k = l+1$.
\begin{align*}
\charp^{(l+1)}(x) & = \frac{\diff}{\dx}\charp^{(l)}(x) = l! \frac{\diff}{\dx} 
		\paren{\sum_{ \substack{ S \subset [n] \\ \Abs{S} = l  } } 
		\charp(x)\Pi_{i \in S} \frac{1}{x - \lambda_i} } 
  = l!  \sum_{ \substack{ S \subset [n] \\ \Abs{S} = l  } } \sum_{j \in [n] \setminus S}
		\paren{\charp(x) \Pi_{i \in S} \frac{1}{x - \lambda_i}} \frac{1}{x - \lambda_j} \\
 & = l! \sum_{ \substack{ S \subset [n] \\ \Abs{S} = l +1}} (l+1)\charp(x) \Pi_{i \in S} \frac{1}{x - \lambda_i} 
	 = (l+1)! \charp(x) \sum_{ \substack{ S \subset [n] \\ \Abs{S} = l +1}} \Pi_{i \in S} \frac{1}{x - \lambda_i} \mper 
\end{align*}
\end{proof}

Next, we analyze $\tg_1(\cdot)$.
\begin{lemma}
\label{lem:tg1}
For $x \in [\lambda_1 + \eprime,1]$, 
$\tg_1(x)$ defined in \prettyref{alg:faster-newton} satisfies 
$ \mom_1(x) \leq \tg_1(x) \leq \mom_1(x) + \delta'  $.
\end{lemma}

\begin{proof}

Using $\xi \defeq 1/(x-\lambda_1)$ for brevity,
\begin{align*}
\frac{\charp(x+\alpha) - \charp(x) }{\alpha} & = \frac{1}{\alpha} \paren{\sum_{j = 0}^\infty 
		\frac{\alpha^j}{j!} \charp^{(j)}(x) - \charp(x)} 
	& \textrm{(Taylor series expansion of $ \charp \paren{\cdot}$ )} \\
 & = \charp'(x) + \frac{1}{\alpha} \paren{\sum_{j = 2}^\infty \alpha^j \charp(x) 
	\sum_{ \substack{ S \subset [n] \\ \Abs{S} = j}} \Pi_{i \in S} \frac{1}{x-\lambda_i} }  
		& \textrm{(Using \prettyref{fact:charpderivative})}\\
 & \leq \charp'(x) + \frac{\charp(x)}{\alpha} \paren{\sum_{j = 2}^\infty \alpha^j n^j \xi^j }
	& \paren{ \textrm{Using} \frac{1}{x - \lambda_i} \leq \xi}	\\
 & \leq \charp'(x) + \charp(x) \frac{\alpha (n \xi)^2 }{1 - \alpha n \xi} \\
 & \leq \charp'(x) + \delta' \charp(x) & \textrm{(Using definition of $\alpha$)}\mper
\end{align*}
Since all the roots of $\charp(x)$ are in $[0,1/2]$ and $x \in [\lambda_1 + \eprime,1]$, 
we have$\charp(x) \leq 1$. Therefore,  
\[ \tg_1(x) = \frac{ \frac{1}{\alpha} \paren{ \charp(x + \alpha) - \charp(x)}}{\charp(x)} 
\leq \mom_1(x) + \delta' \mper \]
Next, since $f(x), x - \lambda_i \geq 0$,  
\[ \frac{\charp(x+\alpha) - \charp(x) }{\alpha} = \charp'(x) + \frac{1}{\alpha} \paren{\sum_{j = 2}^\infty \alpha^j \charp(x) 
	\sum_{ \substack{ S \subset [n] \\ \Abs{S} = j}} \Pi_{i \in S} \frac{1}{x-\lambda_i} }  
		\geq \charp'(x) \mper \]
Therefore,  $\tg_1(x) \geq \mom_1(x)$.

\end{proof}

The crux of the analysis is to show that $\tg_l(x)$ is ``close'' to $\mom_l(x)$.
This is summarised by the following lemma. 
\begin{lemma}[Main Technical Lemma]
\label{lem:approx-der}
For $x \in [\lambda_1 + \eprime,1]$, 
$\tg_{k+1}(x)$ defined in \prettyref{alg:faster-newton} satisfies
 \[ \Abs{ \tg_{k+1}(x) - \mom_{k+1}(x)} \leq \frac{1}{4} \mom_{k+1}(x) \mper \]
\end{lemma}

\begin{proof}
We first bound the quantity $h_{k+1}(x)$ defined as follows.
\begin{align*}
h_{k+1}(x) & \defeq \frac{(-1)^k}{k!\delta^{k}} \paren{ \sum_{i = 0}^{k} (-1)^i \binom{k}{i}  \mom \paren{x + (k-i)\delta }} \\
 & = \frac{(-1)^k}{k!\delta^{k}} \paren{ \sum_{i = 0}^{k} (-1)^i \binom{k}{i} 
		\sum_{j = 0}^\infty \frac{ ((k-i)\delta)^j  }{j!} \mom^{(j)}(x) } 
	& \textrm{(Taylor series expansion of $ \mom \paren{\cdot}$ )} \\
 & = \frac{(-1)^k}{k!\delta^{k}} \paren{ \sum_{j=0}^\infty  \frac{\delta^j \mom^{(j)}(x)}{j!}  
		\sum_{i = 0}^{k} (-1)^i \binom{k}{i} (k-i)^j}  
 	& \textrm{(Rearranging summations)} \\
 & = \frac{(-1)^k}{k!\delta^{k}} \paren{ \sum_{j=0}^\infty  \frac{\delta^j \mom^{(j)}(x)}{j!} 
		(-1)^k \sum_{i = 0}^{k} (-1)^i \binom{k}{i} i^j} 
 	& \textrm{(Rearranging summation)} \\
 & = \mom_{k+1}(x) + \frac{(-1)^k}{k!\delta^{k}} \paren{ \sum_{j=k+1}^\infty  \frac{\delta^j \mom^{(j)}(x)}{j!} 
		(-1)^k \sum_{i = 0}^{k} (-1)^i \binom{k}{i} i^j} & \textrm{(Using \prettyref{fact:sjseries})} \mper
\end{align*}
Using $\xi \defeq 1/(x-\lambda_1)$ for brevity,  
\begin{align*}
\Abs{ h_{k+1}(x) - \mom_{k+1}(x)} & = \Abs{\frac{(-1)^k}{k!} \paren{  \sum_{j=k+1}^\infty  \frac{\delta^{j-k} \mom^{(j)}(x)}{j!} 
		 \sum_{i = 0}^{k} (-1)^i \binom{k}{i} i^j }} \\ 
 & \leq \frac{1}{k!} \paren{  \sum_{j=k+1}^\infty  \delta^{j-k} \mom_{j+1}(x) 
		 \paren{ \sum_{i = 0}^{k}  \binom{k}{i} i^j }} 
		& \textrm{(Using \prettyref{lem:momprime})} \\
 & \leq \frac{1}{k!} \sum_{j=k+1}^\infty \delta^{j-k} \paren{ \mom_{k+1}(x) \xi^{j-k}} k^j 2^k 
	& \textrm{(Using $g_{j+1}(x) \leq \xi^{j-k} g_{k+1}(x)  $)}  \\
 & = \mom_{k+1}(x)  \frac{(2k)^k}{k!} \sum_{p=1}^\infty \paren{k \delta \xi}^{p}  
	& \textrm{(Substituting $p$ for $j-k$) }\\
 & = \mom_{k+1}(x)  \frac{(2k)^k}{k!} \frac{k \delta \xi}{1 - k \delta \xi} \mper
\end{align*}
Next, using \prettyref{lem:tg1} and \prettyref{eq:tgdef}, we have 
\begin{equation} 
\Abs{\tg_{k+1}(x) - h_{k+1}(x)}  \leq \Abs{\frac{1}{k!\delta^{k}} 
		\paren{ \sum_{i = 0}^{k} (-1)^i \binom{k}{i}  \delta'}} 
	\leq \frac{\delta' 2^k}{k! \delta^k} \mper 
\end{equation}

\begin{align*}
\Abs{\tg_{k+1}(x) - \mom_{k+1}(x)} & \leq \Abs{h_{k+1}(x) - \mom_{k+1}(x)} + \Abs{\tg_{k+1}(x) - h_{k+1}(x)} \\
 & \leq \mom_{k+1}(x)  \frac{(2k)^k}{k!} \frac{k \delta \xi}{1 - k \delta \xi} + \frac{\delta' 2^k}{k! \delta^k} \\
 & \leq \mom_{k+1}(x) 4 k \delta \xi \frac{(2k)^k}{k!} & \textrm{(Using $g_{k+1}(x) \geq 1$ and $\delta' = \delta^{k+1} $)} \\
 & \leq \mom_{k+1}(x) (2e)^k 4 k \delta \xi & \textrm{(Using Stirling's approximation for $k!$)} \\
 & \leq \frac{1}{4} \mom_{k+1}(x) \mper
\end{align*}

\end{proof}

\begin{fact}
\label{fact:sjseries}
For $j,k \in \N$,   
\[ \sum_{i = 0}^{k} (-1)^i \binom{k}{i} i^j  = 
	\begin{cases} 0 & \textrm{for } j < k \\ (-1)^k k! & \textrm{for } j = k	\end{cases}  \mper \]
\end{fact}

\begin{proof}
Define the polynomial $S_j(x)$ to be
\begin{equation}
\label{eq:Sjdifferentiation} 
S_j(x) \defeq \underbrace{ x \frac{\diff}{\dx} \ldots x\frac{\diff }{\dx}  }_{j \textrm{ times}} 
		(1+x)^k \mper
\end{equation}
Then,
\[ S_j(x) = \sum_{i = 0}^k \binom{k}{i} i^j x^i  \qquad \textrm{and} \qquad 
 S_j(-1) = \sum_{i = 0}^{k} (-1)^i \binom{k}{i} i^j \mper  \]
Now, for $j < k$, $(1+x)$ will be a factor of the polynomial in \prettyref{eq:Sjdifferentiation}.
Therefore,  $S_j(-1) = 0$ for $j <k$.
For $j=k$, out of the $k+1$ terms of $S_k(x)$ in \prettyref{eq:Sjdifferentiation}, the only term that does not have
a multiple of $(1+x)$ is $x^k k!$. Therefore, $S_k(-1) = (-1)^k k!$.
\end{proof}

Next we show that the update step in \prettyref{alg:faster-newton} (\prettyref{step:faster-newton-update})
makes sufficient progress in each iteration.

\begin{lemma}
\label{lem:noofits2}
For $x^t \in [\lambda_1 + \eprime,1]$,
\[ \frac{x^t - \lambda_1}{8 n^{1/k}}
 \leq \ut  \leq \frac{x^t - \lambda_1}{2} \mper \]
\end{lemma}

\begin{proof}
We first prove the upper bound.
\[ \ut = \frac{1}{4 n^{1/k}}  \frac{ \tg_{k-1}(x^t)}{ \tg_k(x^t)}
\leq  \frac{1}{4n^{1/k}} \frac{(1 + 1/4) \mom_{k-1}(x^t)}{ (1 - 1/4) \mom_k(x^t)}
\leq \frac{x^t - \lambda_1}{2} \mper  \]
Here, the first inequality uses \prettyref{lem:approx-der}  
and the second inequality uses \prettyref{lem:knorm}
with the vector $\paren{ \frac{1}{x^t-\lambda_1}, \ldots, \frac{1}{x^t-\lambda_n} }$.
Next,
\[ \ut = \frac{1}{4n^{1/k}}  \frac{ \tg_{k-1}(x^t)}{ \tg_k(x^t)} 
\geq \frac{1}{4 n^{1/k}} \frac{ (1 - 1/4) \mom_{k-1}(x^t)  }{ (1 + 1/4) \mom_k(x^t) }
\geq \frac{ x^t - \lambda_1 }{8 n^{1/k} } \mper  \]
Here again, the first inequality uses \prettyref{lem:approx-der} 
and the second inequality uses \prettyref{lem:knorm}
with the vector $\paren{ \frac{1}{x^t-\lambda_1}, \ldots, \frac{1}{x^t-\lambda_n} }$.

\end{proof}

\paragraph{Putting it together.}
We now have all the ingredients to prove \prettyref{thm:analysis}. \prettyref{thm:main}
follows from \prettyref{thm:analysis} by picking $k = \log n$.
\begin{proof}[Proof of \prettyref{thm:analysis}]
We first analyze the output guarantees of \prettyref{alg:faster-newton}, and then we bound its bit complexity. 

\paragraph{Invariants and Output Guarantees.}
W.l.o.g., we may assume that $x^0 - \lambda_1 \geq \eprime$.
We will assume that $x^t - \lambda_1 \geq \eprime$, and show that if the algorithm does not
stop in this iteration, then $x^{t+1} - \lambda_1 \geq \eprime$, thereby justifying our assumption.

Since we do \prettyref{step:update} only when $\ut > \eprime$, we get using \prettyref{lem:noofits2} 
that 
\[ \tut \geq \ut - \frac{\eprime}{n} > \frac{\ut}{2} \geq \frac{x^t - \lambda_1}{16 n^{1/k}} \mper \]
Using \prettyref{lem:noofits1},  we get that for some iteration 
$t \leq 16 n^{1/k} \log (1/\e)$ of \prettyref{alg:faster-newton}, 
we will have $x^t - \lambda_1 \leq \e$.
Therefore, if the algorithm does not stop at \prettyref{step:stop}, and 
terminates at \prettyref{step:endofalg}, the $\lambda$ output by the algorithm
will satisfy $0 \leq \lambda - \lambda_1 \leq \e$.
If the algorithm does stop at \prettyref{step:stop}, i.e., 
$\ut \leq \eprime$, then from \prettyref{lem:noofits2} we get
\[ x^t - \lambda_1 \leq 8 n^{1/k}\, \ut \leq 8 n^{1/k}\, \eprime  \leq \e \mper  \]
Therefore, in both these cases, the algorithm outputs a $\lambda$ satisfying 
$0 \leq \lambda - \lambda_1 \leq \e$.

Next, if the algorithm does not stop in \prettyref{step:stop}, then we get from
\prettyref{lem:noofits2} that  
\begin{equation}
\label{eq:nostop}
 \eprime <  \ut \leq \frac{x^t - \lambda_1}{2} 
\end{equation}
and since $\tut \leq \ut$,  
\begin{align*}
x^{t+1} - \lambda_1 & = \paren{x^t - \tut} - \lambda_1 \geq \paren{x^t - \ut} - \lambda_1 
				& \textrm{(from \prettyref{step:faster-newton-update} of 
			\prettyref{alg:faster-newton})} \\
 & \geq  \frac{x^t - \lambda_1}{2} & \textrm{(from \prettyref{lem:noofits2})} \\            
 & > \eprime   & \textrm{(from \prettyref{eq:nostop})}\mper
\end{align*}
Therefore, if we do not stop in interation $t$, then we ensure that  
$ x^{t+1} - \lambda_1 \geq \eprime$.

\paragraph{Bit Complexity.}
We now bound the number of bit operations performed by the algorithm. 
We will show by induction on $t$ that the bit complexity of each $x^t$ is 
\begin{equation}
\label{eq:bitcompxt}
\bitcomp{x^t} \leq \log (n/\eprime) \mper
\end{equation}
We will assume that $\bitcomp{x^t} \leq \log (n/\eprime)$. We use this to
bound the number of bit operations performed in each step of \prettyref{alg:faster-newton} 
and to show that $\bitcomp{x^{t+1}} \leq \log (n/\eprime)$.

Each computation of $\tg_1(\cdot)$ involves two computations of $\charp(\cdot)$
and one division by $\charp(\cdot)$. The bit complexity of the locations at which 
$\charp(\cdot)$ is computed can be upper bounded by 
\[ \bitcomp{x^t} + \bitcomp{k \delta} + \bitcomp{\alpha} = \bigo{\bitcomp{\alpha}}  \mper    \]
From our assumption that $\ftime{n_1} \leq \ftime{n_2},\ \forall n_1 \leq n_2$, 
and that $\ftime{cn} = \bigo{\ftime{n}}$, 
we get that 
the bit complexity of each of these $\charp(\cdot)$ computations can be bounded by
$\bigo{\ftime{\bitcomp{\alpha}}}$.
Since, division can be done in nearly linear time \cite{ss71},
the bit complexity of the computation of $\tg_1(\cdot)$ is 
$\tbigo{\ftime{\bitcomp{\alpha}}}$.

The computation of the $\tg_k(\cdot)$ involves $k$ computations of $\tg_1(\cdot)$ and
one division by $\delta^k$,
and therefore can be done using $\tbigo{k \ftime{\bitcomp{\alpha}}}$ bit operations.
Next, the computation of $u^t$ involves computing the ratio of $\tg_{k-1}(x^t)$ and $\tg_k(x^t)$,  
both of which have bit complexity $\tbigo{k \ftime{\bitcomp{\alpha}}}$. 
Therefore, $u^t$ can be computed in $\tbigo{k \ftime{\bitcomp{\alpha}}}$
bit operations \cite{ss71}.
Finally, since $x^{t+1} = x^t - \tut$, we get that 
$\bitcomp{x^{t+1}} =  \bitcomp{x^t - \tut} \leq \log (n/\eprime)$.
For our choice of parameters
\[ \bitcomp{\alpha} = \log \paren{\frac{2n^2 (16k)^{k+1} (2e)^{k^2+k}}{\eprime^{k+3}} } 
 = \bigo{k^2 + k\log (n/\e)} \mper \]

Finally, since the number of iterations in the algorithm is at most $16 n^{1/k}  \log \paren{1/\e }$,
the overall query complexity of the algorithm is $\bigo{ n^{1/k}  \log \paren{1/\e } \cdot k}$, 
and the overall bit complexity (running time) is 
\[\tbigo{  n^{1/k} \log \paren{1/\e }\cdot k \ftime{ k^2 +k \log (n/\e) }} \mper  \]
\end{proof}

\subsection{Computing the top eigenvalue of a matrix}
\label{sec:matrix}

Our algorithm (\prettyref{thm:main}) uses an algorithm the compute the determinant of a matrix as a
subroutine.
Computing the determinant of a matrix has many applications in theoretical computer science
and is a well studied problem. We refer the reader to \cite{kv04} for a survey.
The algorithm for computing the determinant of a matrix with the current fastest asymptotic running time 
is due to Storjohann \cite{s05}.
\begin{theorem}[\cite{s05}]
\label{thm:fast-determinant}
Let $A \in \Z^{n \times n}$. There exists a Las Vegas algorithm that computes $\det (A)$
using an expected number of $\bigo{ n^{\omega} \log^2 n \log \fnorm{A} }$ bit operations.
\end{theorem}

\begin{proof}[Proof of \prettyref{thm:main}]
Using \prettyref{thm:fast-determinant}, each computation of $\charp(x) = \det(xI - A)$ can be done in time  
\[ \bigo{n^{\omega} \log^2n \log \paren{\fnorm{A}/\alpha}}  
	= \bigo{n^{\omega} \log^2n \paren{ k^2 \log \paren{ \fnorm{A}/\e }}} \mper \]
Using \prettyref{thm:analysis} with $k = \ceil{\log n}$, the overall bit complexity (running time) of \prettyref{alg:faster-newton} is 
\[ \bigo{ n^{\omega } \log^{5} n \log^2 \paren{\fnorm{A}/\e }}  \mper  \]
\end{proof}

\paragraph{Acknowledgement.} We are grateful to Ryan O' Donnell for helpful discussions, 
and to Yin Tat Lee, Prasad Raghavendra, Aaron Schild and Aaron Sidford for pointing us to the 
finite difference method for approximating higher derivatives efficiently. 

\bibliography{bibfile,bibfileacm}
\bibliographystyle{amsalpha}
\end{document}

%% file: macros.tex
\usepackage{amsmath,amstext,amssymb,amsfonts}
\usepackage{color,graphicx}
\usepackage{tabularx}
\usepackage{verbatim}

\usepackage{amsthm}
\usepackage{ifdraft}

\usepackage{nicefrac}

\usepackage{microtype}

\newtheorem{theorem}{Theorem}[section]
\newtheorem*{theorem*}{Theorem}

\newtheorem*{claim*}{Claim}

\newtheorem{proposition}[theorem]{Proposition}
\newtheorem*{proposition*}{Proposition}
\newtheorem{lemma}[theorem]{Lemma}
\newtheorem*{lemma*}{Lemma}
\newtheorem{corollary}[theorem]{Corollary}

\newtheorem*{conjecture*}{Conjecture}

\newtheorem{fact}[theorem]{Fact}
\newtheorem*{fact*}{Fact}

\newtheorem*{hypothesis*}{Hypothesis}

\theoremstyle{definition}

\newtheorem{algorithm}[theorem]{Algorithm}

\newtheorem{assumption}[theorem]{Assumption}

\usepackage{prettyref}
\newcommand{\savehyperref}[2]{\texorpdfstring{\hyperref[#1]{#2}}{#2}}

\newrefformat{eq}{\savehyperref{#1}{\textup{(\ref*{#1})}}}
\newrefformat{lem}{\savehyperref{#1}{Lemma~\ref*{#1}}}
\newrefformat{def}{\savehyperref{#1}{Definition~\ref*{#1}}}
\newrefformat{thm}{\savehyperref{#1}{Theorem~\ref*{#1}}}
\newrefformat{cor}{\savehyperref{#1}{Corollary~\ref*{#1}}}
\newrefformat{cha}{\savehyperref{#1}{Chapter~\ref*{#1}}}
\newrefformat{sec}{\savehyperref{#1}{Section~\ref*{#1}}}
\newrefformat{app}{\savehyperref{#1}{Appendix~\ref*{#1}}}
\newrefformat{tab}{\savehyperref{#1}{Table~\ref*{#1}}}
\newrefformat{fig}{\savehyperref{#1}{Figure~\ref*{#1}}}
\newrefformat{hyp}{\savehyperref{#1}{Hypothesis~\ref*{#1}}}
\newrefformat{alg}{\savehyperref{#1}{Algorithm~\ref*{#1}}}
\newrefformat{sdp}{\savehyperref{#1}{SDP~\ref*{#1}}}
\newrefformat{rem}{\savehyperref{#1}{Remark~\ref*{#1}}}
\newrefformat{item}{\savehyperref{#1}{Item~\ref*{#1}}}
\newrefformat{step}{\savehyperref{#1}{step~\ref*{#1}}}
\newrefformat{conj}{\savehyperref{#1}{Conjecture~\ref*{#1}}}
\newrefformat{fact}{\savehyperref{#1}{Fact~\ref*{#1}}}
\newrefformat{prop}{\savehyperref{#1}{Proposition~\ref*{#1}}}
\newrefformat{claim}{\savehyperref{#1}{Claim~\ref*{#1}}}
\newrefformat{relax}{\savehyperref{#1}{Relaxation~\ref*{#1}}}
\newrefformat{red}{\savehyperref{#1}{Reduction~\ref*{#1}}}
\newrefformat{part}{\savehyperref{#1}{Part~\ref*{#1}}}
\newrefformat{prob}{\savehyperref{#1}{Problem~\ref*{#1}}}
\newrefformat{ass}{\savehyperref{#1}{Assumption~\ref*{#1}}}

\newcommand{\Sref}[1]{\hyperref[#1]{\S\ref*{#1}}}

\usepackage[varg]{txfonts}    

\renewcommand{\mathbb}{\varmathbb} 

\renewcommand{\leq}{\leqslant}
\renewcommand{\le}{\leqslant}
\renewcommand{\geq}{\geqslant}

\usepackage{bm}

\usepackage{xspace}

\usepackage[pdftex,pagebackref,colorlinks,linkcolor=blue,filecolor = blue, citecolor = blue, urlcolor  = blue]{hyperref}

\usepackage{boxedminipage}

\newenvironment{mybox}
{\center \noindent\begin{boxedminipage}{1.0\linewidth}}
{\end{boxedminipage}
\noindent
}

\newcommand{\mper}{\,.}

\newcommand{\paren}[1]{\left(#1 \right )}

\newcommand{\set}[1]{\left\{#1\right\}}

\newcommand{\Abs}[1]{\left\lvert#1\right\rvert}

\newcommand{\ceil}[1]{\lceil #1 \rceil}

\newcommand{\norm}[1]{\left\lVert#1\right\rVert}

\newcommand{\fnorm}[1]{\norm{#1}_F}

\newcommand{\defeq}{\stackrel{\textup{def}}{=}}

\newcommand{\normi}[1]{\norm{#1}_{\scriptstyle \infty}}

\newcommand{\Z}{{\mathbb Z}}
\newcommand{\N}{{\mathbb Z}_{\geq 0}}
\newcommand{\R}{\mathbb R}
\newcommand{\Q}{\mathbb Q}

\newcommand{\e}{\epsilon}

\definecolor{DSgray}{cmyk}{0,0,0,0.7}

\let\e\varepsilon

\newcommand{\cB}{\mathcal B}

\newcommand{\cT}{\mathcal T}

\newcommand{\etal}{et. al.}
\newcommand{\bigO}{\mathcal{O}}
\newcommand{\bigo}[1]{\bigO\left(#1\right)}
\newcommand{\tbigO}{\tilde{\mathcal{O}}}
\newcommand{\tbigo}[1]{\tbigO\left(#1\right)}



%% file: computingeigs-arxiv2.bbl
\providecommand{\bysame}{\leavevmode\hbox to3em{\hrulefill}\thinspace}
\providecommand{\MR}{\relax\ifhmode\unskip\space\fi MR }
\providecommand{\MRhref}[2]{%
  \href{http://www.ams.org/mathscinet-getitem?mr=#1}{#2}
}
\providecommand{\href}[2]{#2}
\begin{thebibliography}{KKZN97}

\bibitem[Arn51]{a51}
Walter~Edwin Arnoldi, \emph{The principle of minimized iterations in the
  solution of the matrix eigenvalue problem}, Quarterly of Applied Mathematics
  \textbf{9} (1951), no.~1, 17--29.

\bibitem[Bar68]{bar68}
Erwin~H Bareiss, \emph{Sylvester's identity and multistep integer-preserving
  gaussian elimination}, Mathematics of computation \textbf{22} (1968),
  no.~103, 565--578.

\bibitem[Bla66]{b66}
WA~Blankinship, \emph{Matrix triangulation with integer arithmetic},
  Communications of the ACM \textbf{9} (1966), no.~7, 513.

\bibitem[BOE15]{be15}
Michael Ben-Or and Lior Eldar, \emph{The quasi-random perspective on matrix
  spectral analysis with applications}, arXiv preprint arXiv:1505.08126 (2015).

\bibitem[DDH07]{ddh07}
James Demmel, Ioana Dumitriu, and Olga Holtz, \emph{Fast linear algebra is
  stable}, Numerische Mathematik \textbf{108} (2007), no.~1, 59--91.

\bibitem[Dix82]{dix82}
John~D Dixon, \emph{Exact solution of linear equations usingp-adic expansions},
  Numerische Mathematik \textbf{40} (1982), no.~1, 137--141.

\bibitem[Edm67]{e67}
Jack Edmonds, \emph{Systems of distinct representatives and linear algebra}, J.
  Res. Nat. Bur. Standards, Sect. B \textbf{71} (1967), no.~4, 241--245.

\bibitem[EH99]{eh99}
J.~A. Ezquerro and M.~A. Hern\'{a}ndez, \emph{On a convex acceleration of
  newton's method}, J. Optim. Theory Appl. \textbf{100} (1999), no.~2,
  311--326.

\bibitem[FP96]{fp96}
William~F. Ford and James~A. Pennline, \emph{Accelerated convergence in
  newton's method}, SIAM Rev. \textbf{38} (1996), no.~4, 658--659.

\bibitem[Fru77]{f77}
Michael~A Frumkin, \emph{Polynomial time algorithms in the theory of linear
  diophantine equations}, Fundamentals of Computation Theory, Springer, 1977,
  pp.~386--392.

\bibitem[Ger94]{g94}
J\"{u}rgen Gerlach, \emph{Accelerated convergence in newton's method}, SIAM
  Rev. \textbf{36} (1994), no.~2, 272--276.

\bibitem[GH01]{gh01}
J.~M. Guti{\'e}rrez and M.~A. Hern\'{a}ndez, \emph{An acceleration of newton's
  method: Super-halley method}, Appl. Math. Comput. \textbf{117} (2001),
  no.~2-3, 223--239.

\bibitem[Hou70]{h70}
Alston~Scott Householder, \emph{The numerical treatment of a single nonlinear
  equation}, McGraw-Hill New York, 1970.

\bibitem[KG00]{kg00}
Bahman Kalantari and J\"{u}rgen Gerlach, \emph{Newton's method and generation
  of a determinantal family of iteration functions}, J. Comput. Appl. Math.
  \textbf{116} (2000), no.~1, 195--200.

\bibitem[KKZN97]{kkz97}
Bahman Kalantari, Iraj Kalantari, and Rahim Zaare-Nahandi, \emph{A basic family
  of iteration functions for polynomial root finding and its
  characterizations}, J. Comput. Appl. Math. \textbf{80} (1997), no.~2,
  209--226.

\bibitem[KLW06]{klw06}
Jisheng Kou, Yitian Li, and Xiuhua Wang, \emph{A modification of newton method
  with third-order convergence}, Appl. Math. Comput. \textbf{181} (2006),
  no.~2, 1106--1111.

\bibitem[KV04]{kv04}
Erich Kaltofen and Gilles Villard, \emph{Computing the sign or the value of the
  determinant of an integer matrix, a complexity survey}, Journal of
  Computational and Applied Mathematics \textbf{162} (2004), no.~1, 133--146.

\bibitem[KV05]{kv05b}
\bysame, \emph{On the complexity of computing determinants}, Computational
  complexity \textbf{13} (2005), no.~3-4, 91--130.

\bibitem[Lan50]{l50}
Cornelius Lanczos, \emph{An iteration method for the solution of the eigenvalue
  problem of linear differential and integral operators}, United States
  Governm. Press Office, 1950.

\bibitem[LR08]{lr08}
Tibor Luki{\'c} and Neboj{\v{s}}a~M Ralevi{\'c}, \emph{Geometric mean newton's
  method for simple and multiple roots}, Applied Mathematics Letters
  \textbf{21} (2008), no.~1, 30--36.

\bibitem[MPG29]{mp29}
RV~Mises and Hilda Pollaczek-Geiringer, \emph{Praktische verfahren der
  gleichungsaufl{\"o}sung.}, ZAMM-Journal of Applied Mathematics and
  Mechanics/Zeitschrift f{\"u}r Angewandte Mathematik und Mechanik \textbf{9}
  (1929), no.~1, 58--77.

\bibitem[NH13]{nh13}
Yuji Nakatsukasa and Nicholas~J Higham, \emph{Stable and efficient spectral
  divide and conquer algorithms for the symmetric eigenvalue decomposition and
  the svd}, SIAM Journal on Scientific Computing \textbf{35} (2013), no.~3,
  A1325--A1349.

\bibitem[OR70]{or70}
James~M Ortega and Werner~C Rheinboldt, \emph{Iterative solution of nonlinear
  equations in several variables}, vol.~30, Siam, 1970.

\bibitem[OW08]{ow08}
Christina Oberlin and Stephen~J. Wright, \emph{An accelerated newton method for
  equations with semismooth jacobians and nonlinear complementarity problems},
  Math. Program. \textbf{117} (2008), no.~1, 355--386.

\bibitem[Pan96]{p96}
Victor~Y Pan, \emph{Optimal and nearly optimal algorithms for approximating
  polynomial zeros}, Computers \& Mathematics with Applications \textbf{31}
  (1996), no.~12, 97--138.

\bibitem[Pan97]{p97}
\bysame, \emph{Solving a polynomial equation: some history and recent
  progress}, SIAM review \textbf{39} (1997), no.~2, 187--220.

\bibitem[PC99]{pc99}
Victor~Y Pan and Zhao~Q Chen, \emph{The complexity of the matrix eigenproblem},
  Proceedings of the thirty-first annual ACM symposium on Theory of computing,
  ACM, 1999, pp.~507--516.

\bibitem[PTVF92]{ptvf92}
William~H Press, Saul~A Teukolsky, William~T Vetterling, and Brian~P Flannery,
  \emph{Numerical recipes (cambridge}, 1992.

\bibitem[Rut71]{r71}
Heinz Rutishauser, \emph{The jacobi method for real symmetric matrices}, Linear
  algebra, Springer, 1971, pp.~202--211.

\bibitem[Sch98]{s98}
Alexander Schrijver, \emph{Theory of linear and integer programming}, John
  Wiley \& Sons, 1998.

\bibitem[SS71]{ss71}
A.~Sch{\"o}nhage and V.~Strassen, \emph{Schnelle multiplikation grosser
  zahlen}, Computing \textbf{7} (1971), no.~3-4, 281--292.

\bibitem[Sto05]{s05}
Arne Storjohann, \emph{The shifted number system for fast linear algebra on
  integer matrices}, Journal of Complexity \textbf{21} (2005), no.~4, 609--650.

\bibitem[Vis13]{v13}
Nisheeth~K. Vishnoi, \emph{Lx = b}, Foundations and Trends® in Theoretical
  Computer Science \textbf{8} (2013), 1--141.

\end{thebibliography}
